\documentclass{article}

\usepackage{color}
\usepackage{amsmath, amssymb}

\usepackage{amsthm}
\usepackage{mathtools}  
\usepackage{physics}    
\usepackage{nicematrix} 

\newtheorem{defi}{Definition}

\newtheorem{prop}{Proposition} 

\newtheorem{thm}{Theorem}[section]

\begin{document}

\title{{\bf Absolute zeta functions and periodicity \\ 
of quantum walks on cycles}
\vspace{15mm}}

\author{Jir\^o AKAHORI$^{1}$, $\quad$ Norio KONNO$^{2, \ast}$, 
\\
\\
Iwao SATO$^{3}$, $\quad$ Yuma TAMURA$^{4}$,
\\ 
\\ 
\\
Department of Mathematical Sciences \\
College of Science and Engineering \\
Ritsumeikan University \\
1-1-1 Noji-higashi, Kusatsu, 525-8577, JAPAN \\
e-mail: akahori@se.ritsumei.ac.jp$^{1}$, \ n-konno@fc.ritsumei.ac.jp$^{2,\ast},$\\
ytamura11029@gmail.com$^{4}$ 
\\
\\
Oyama National College of Technology \\
771 Nakakuki, Oyama 323-0806, JAPAN \\
e-mail: isato@oyama-ct.ac.jp$^{3}$
}

\date{\empty }

\maketitle

\vspace{50mm}


\vspace{20mm}

\begin{small}
\par\noindent
{\bf Corresponding author}: Yuma Tamura, Department of Mathematical Sciences, College of Science and Engineering, Ritsumeikan University, 1-1-1 Noji-higashi, Kusatsu, 525-8577, JAPAN \\
e-mail: ytamura11029@gmail.com
\par\noindent
\\
\par\noindent
{\bf Abbr. title: Absolute zeta functions and periodicity of quantum walks} 

\end{small}

\clearpage

\begin{abstract}
The quantum walk is a quantum counterpart of the classical random walk. On the other hand, absolute zeta functions can be considered as zeta functions over $\mathbb{F}_1$. This study presents a connection between quantum walks and absolute zeta functions. In this paper, we focus on Hadamard walks and $3$-state Grover walks on cycle graphs. The Hadamard walks and the Grover walks are typical models of the quantum walks. We consider the periods and zeta functions of such quantum walks. Moreover, we derive the explicit forms of the absolute zeta functions of corresponding zeta functions. Also, it is shown that our zeta functions of quantum walks are absolute automorphic forms. 
\end{abstract}

\vspace{10mm}

\begin{small}
\end{small}

\vspace{10mm}

\section{Introduction \label{sec01}}
This work is a continuation of \cite{AkahoriEtAL2023,Konno2023}. Quantum walks are considered to be the corresponding model for random walks in quantum systems. Quantum walks play important roles in various fields such as mathematics, quantum physics, and quantum information processing. Concerning quantum walks, see \cite{GZ, Konno2008, ManouchehriWang, Portugal, Venegas}, and as for random walks, see \cite{Norris, Spitzer}, for instance. On the other hand, absolute zeta functions are zeta functions over $\mathbb{F}_1$, where $\mathbb{F}_1$ can be viewed as a kind of limit of $\mathbb{F}_p$ as $p \to 1$. Here $\mathbb{F}_p = \mathbb{Z}/p \mathbb{Z}$ stands for the field of $p$ elements for a prime number $p$. This paper presents a connection between quantum walks and absolute zeta functions. Concerning absolute zeta functions, see \cite{CC, KF1, Kurokawa3, Kurokawa, KO, KT3, KT4, Soule}.

In this study, we deal with Hadamard walks and $3$-state Grover walks on cycle graphs. More precisely, we study the periods and zeta functions of such quantum walks. We provide simple and unified proof to determine the periods, although some of the results are already known. Moreover, we calculate the absolute zeta functions of Hadamard walks and Grover walks in the case where the period is finite. The explicit forms of the absolute zeta functions can be seen in this paper. It is also shown that the zeta functions of such quantum walks are absolute automorphic forms. 

The rest of this paper is organized as follows. Section \ref{sec02} briefly describes absolute zeta functions and its some related topics. In Section \ref{CycleGraphs}, we introduce cycle graphs and quantum walks on them. Section \ref{CyclotomicPoly} is a brief check of the notion of cycrotomic polynomial. Section \ref{Hadamard} explains the period and zeta functions of Hadamard walks on cycle graphs, and also absolute zeta function of them. In Section \ref{Grover}, we treat Grover walks with $3$ states. Finally, Section \ref{conclusion} is devoted to conclusion.

\section{Absolute Zeta Function \label{sec02}}
First we introduce the following notation: $\mathbf{Z}$ is the set of integers, $\mathbf{Z}_{>0} = \{1,2,3, \ldots \}$,  $\mathbf{R}$ is the set of real numbers, and $\mathbf{C}$ is the set of complex numbers. 

In this section, we briefly review the framework on absolute zeta functions, which can be considered as zeta function over $\mathbf{F}_1$, and absolute automorphic forms (see \cite{Kurokawa3, Kurokawa, KO, KT3, KT4} and references therein, for example). 

Let $f(x)$ be a function $f : \mathbf{R} \to \mathbf{C} \cup \{ \infty \}$. We say that $f$ is an {\em absolute automorphic form} of weight $D$ if $f$ satisfies
\begin{align*}
f \left( \frac{1}{x} \right) = C x^{-D} f(x)
\end{align*}
with $C \in \{ -1, 1 \}$ and $D \in \mathbf{Z}$. The {\em absolute Hurwitz zeta function} $Z_f (w,s)$ is defined by
\begin{align*}
Z_f (w,s) = \frac{1}{\Gamma (w)} \int_{1}^{\infty} f(x) \ x^{-s-1} \left( \log x \right)^{w-1} dx,
\end{align*}
where $\Gamma (x)$ is the gamma function (see \cite{Andrews1999}, for instance). Then taking $u=e^t$, we see that $Z_f (w,s)$ can be rewritten as the Mellin transform: 
\begin{align*}
Z_f (w,s) = \frac{1}{\Gamma (w)} \int_{0}^{\infty} f(e^t) \ e^{-st} \ t^{w-1} dt.
\end{align*}
Moreover, the {\em absolute zeta function} $\zeta_f (s)$ is defined by 
\begin{align*}
\zeta_f (s) = \exp \left( \frac{\partial}{\partial w} Z_f (w,s) \Big|_{w=0} \right).
\end{align*}
Here we introduce the {\em multiple Hurwitz zeta function of order $r$}, $\zeta_r (s, x, (\omega_1, \ldots, \omega_r))$, the {\em multiple gamma function of order $r$}, $\Gamma_r (x, (\omega_1, \ldots, \omega_r))$, and the {\em multiple sine function of order $r$}, $S_r (x, (\omega_1, \ldots, \omega_r))$, respectively (see \cite{Kurokawa3, Kurokawa, KT3}, for example): 
\begin{align*}
\zeta_r (s, x, (\omega_1, \ldots, \omega_r))
&= \sum_{n_1=0}^{\infty} \cdots \sum_{n_r=0}^{\infty} \left( n_1 \omega_1 + \cdots + n_r \omega_r + x \right)^{-s}, 
\\
\Gamma_r (x, (\omega_1, \ldots, \omega_r)) 
&= \exp \left( \frac{\partial}{\partial s} \zeta_r (s, x, (\omega_1, \ldots, \omega_r)) \Big|_{s=0} \right),
\\
S_r (x, (\omega_1, \ldots, \omega_r))
&= \Gamma_r (x, (\omega_1, \ldots, \omega_r))^{-1} \ \Gamma_r (\omega_1+ \cdots + \omega_r - x, (\omega_1, \ldots, \omega_r))^{(-1)^r}.
\end{align*}
\par
Now we present the following key result derived from Theorem 4.2 and its proof in Korokawa \cite{Kurokawa} (see also Theorem 1 in Kurokawa and Tanaka \cite{KT3}):

\begin{thm}\label{ExplicitAZeta}
    If \(f\) has the form
    \[
        f(x) = x^{l/2} \frac{ ( x^{m(1)}-1 ) \cdots ( x^{m(a)}-1 ) }{ ( x^{n(1)}-1 ) \cdots ( x^{n(b)}-1 ) }
    \]
    for some $ l \in \mathbf{Z} $, $ a,b \in \mathbf{Z}_{>0} $, $ m(i), n(j) \in \mathbf{Z}_{>0}\, (i=1,\dots,a, j=1,\dots,b) $, then the following holds:
    \begin{align*}
        Z_f(w,s) &= \sum_{ I \subset \{ 1,\dots,a \} } (-1)^{|I|} \zeta_b( w, s-\deg(f) + m(I), \mbox{\boldmath $n$} ),   \\
        \zeta_f(s) &= \prod_{ I \subset \{ 1,\dots,a \} } \Gamma_b( s - \deg(f) + m(I), \mbox{\boldmath $n$} )^{ (-1)^{ |I| } },    \\
        \zeta_f( D-s )^C &= \varepsilon_f(s) \zeta_f(s),
    \end{align*}
    where
    \begin{gather*}
        \deg(f) = l/2 + \sum_{i=1}^a m(i) - \sum_{j=1}^b n(j), \quad m(I) = \sum_{ i \in I } m(i), \\
        \mbox{\boldmath $n$} = ( n(1), \dots, n(j) ), \quad D = l + \sum_{i=1}^a m(i) - \sum_{j=1}^b n(j), \\
        C = (-1)^{a-b}, \quad \varepsilon_f= \prod_{ I \subset \{ 1,\dots,a \} } S_b( s - \deg(f) + m(I), \mbox{ \boldmath $n$ } ).
    \end{gather*}
\end{thm}

\section{Quantum walks on Cycle graphs}\label{CycleGraphs}

For $ N \ge 2 $, \emph{undirected cycle graph with $N$ vertices} is an undirected graph which has $ N $ vertices and each vertex connecting to exactly $2$ edges. We write this graph by $C_N$. Formally, $ C_N $ is defined in the following way:

\begin{defi}
    The set of vertices of $ C_N $ is $ \{ 0, 1, \dots, N-1 \} $ and the set of edges of $ C_N $ is $ \{ \{ k, k+1 \} \mid k=0,\dots,N-1 \} $, where the numbers are identified modulo $N$.
\end{defi}
%
A quantum walk is the time-evolving sequence of states consisting of position and chirality. Formally, a state is a vector which is an element of a tensor product of two Hilbert spaces over $ \mathbf{C} $, $ \mathcal{H}_P $ and $ \mathcal{H}_C  $. In this paper, $ \mathcal{H}_P $ is a vector space over $ \mathbf{C} $ in which $ \{ \ket{x} \mid x \in V(C_N) \} $ is orthonormal basis, where $V(C)$ is the set of vertices of $C_N$. Also, $ \mathcal{H}_C $ is a vector space over $ \mathbf{C} $ in which $ \{ \ket{ \leftarrow }, \ket{ \rightarrow } \} $ is an orthonormal basis for Hadamard walks, and in which $ \{ \ket{ \leftarrow }, \ket{\cdot}, \ket{ \rightarrow } \} $ is an orthonormal basis for Grover walks with $3$ states. Note that the elements of $ \mathcal{H}_P $ and $ \mathcal{H}_C $ are considered to be the orthonormal basis of each space. Then, each state can be represented like the following:
\[
    \sum_{ x \in V(C_n) } \ket{x} \otimes s, \qquad s \in \mathcal{H}_C.
\]
Usually, we assume that the initial state $ \Psi_0 $ satisfies $ \| 
\Psi_0 \| = 1 $. Moreover, we consider the case where the time-evolution operator $ U $ is decomposed as $ U = SC $. Here, $ S $ is called \emph{shift operator} and defined by the following formulas:
\begin{align*}
    S( \ket{x} \otimes \ket{ \leftarrow } ) &:= \ket{ x-1 } \otimes \ket{ \leftarrow }, \\
    S( \ket{x} \otimes \ket{ \rightarrow }) &:= \ket{ x+1 } \otimes \ket{ \rightarrow },   \\
    S( \ket{x} \otimes \ket{ \cdot }) &:= \ket{ x } \otimes \ket{ \cdot }.
\end{align*}
Furthermore, $ C $ is called \emph{coin operator} and defined by the following:
\[
    C := \sum_{ x \in V(C_N) } \ketbra{x}{x} \otimes A
\]
for some unitary operator $A$ on $ \mathcal{H}_C $. We call this operator $ A $ \emph{the local coin operator}. In this case, $ S $ and $ C $ are both unitary, and then $ U $ is also unitary. Now, the time-evolution is defined as usual:
\[
    \Psi_{n+1} := U \Psi_n.
\]
Of course, we have $ \Psi_n = U^n \Psi_0 $. We are interested in this time-evolution operator $U$. In each of the subsequent subsections, matrix representations of $ U $ are shown. Moreover, we introduce the period of quantum walk.
\begin{defi}
    For a quantum walk whose time-evolution operater is $U$, \emph{the period of the quantum walk} is defined as the infimum:
    \[
        \inf \{ n \ge 1 \mid U^n = 1 \}.
    \]
    If the set in the above formula is empty, then the period is defined to be $ \infty $.
\end{defi}
Of course, if $ T $ is the period of a quantum walk, it holds that
\[
    \Psi_T = \Psi_0
\]
for any initial state $ \Psi_0 $.

\section{Cyclotomic polynomials}\label{CyclotomicPoly}
Also, we treat polynomial rings and cyclotomic polynomials in this paper.
\begin{defi}
    $ \mathbf{Z}[x] $ and $ \mathbf{Q}[x] $ denote the polynomial rings with integer and rational coefficients, respectively.
\end{defi}
Then \emph{cyclotomic polynomials} are defined as follows:
\begin{defi}
    For $ n \in \mathbf{Z}_{>0} $, \emph{cyclotonomic polynomial} $ \Phi_n(x) $ is defined by the following formula:
    \[
        \Phi_n(x) := \prod_{ \substack{ 1 \le k \le n-1 \\ \gcd(k,n) = 1 } } \biggl( x - \exp\biggl( \frac{ 2 \pi i k}{ n } \biggr) \biggr).
    \]
\end{defi}
Note that $ \Phi_n( x ) \in \mathbf{Z}[x] $ for all $n$. Now, the subsequent proposition is the key of this paper.
\begin{prop}[See e.g. \cite{HiguchiEtAl2017}]\label{monic_unity}
    If all of the roots of a monic polynomial with rational coefficients $ f(x) $ are roots of unity, then $ f(x) \in \mathbf{Z}[x] $.
\end{prop}
Here, \emph{monic} means ``the nonzero coefficient of highest degree is equal to $1$,'' and \emph{root of unity} means a complex number $ z $ which satisfies
\[
    z^n = 1
\]
for some $ n \in \mathbf{Z}_{>0} $. This proposition is a consequence of the fact that the minimal polynomial over $ \mathbf{Q} $ of any root of unity is cyclotonomic polynomial, in particular, integer coefficient polynomial.

In this paper, this proposition is used in the following manner. First, note that for any square matrix $ A $ and positive integer $n$, if $ \lambda $ is an eigenvalue of $ A $, then $ \lambda^n $ is an eigenvalue of $ A^n $. Therefore, if $ A $ has a complex number which is not root of unity as its eigenvalue, then every $ A^n $ has a complex number which is not equal to $ 1 $ as its eigenvalue. This implies that $ A^n $ is not the identity matrix.

\section{Hadamard walks}\label{Hadamard}

Hadamard walks are a well-studied class of quantum walks. There are two types of Hadamard walks: M type and F type. They are characterized by local coin operators as usual. First, we introduce the definition:

\begin{defi}
    \emph{M-type} (respectively \emph{F-type}) \emph{Hadamard walks} on $ C_N $ are the quantum walks whose local coin operators with respect to the ordered basis $ ( \ket{\leftarrow}, \ket{\rightarrow} ) $ are as follows respectively:
    \begin{align*}
        A^{H,M} = \frac{1}{ \sqrt{2} }\,
        \begin{bmatrix}
            1   &   1   \\
            1   &  -1
        \end{bmatrix},
        &&  A^{H,F} = \frac{1}{ \sqrt{2} }\,
        \begin{bmatrix}
            1   &  -1   \\
            1   &   1
        \end{bmatrix}.
    \end{align*}
\end{defi}

In this paper, we focus on Hadamard walks of F type. $ U^{H,F}_N $ denotes the time-evolution operator of an F-type Hadamard walk on $ C_N $. With respect to the ordered basis of $ \mathcal{H}_P \otimes \mathcal{H}_C $ $ ( \bra{0} \otimes \bra{\leftarrow}, \bra{0} \otimes \bra{\rightarrow}, \bra{1} \otimes \bra{\leftarrow}, \dots, \bra{N-1} \otimes \bra{ \rightarrow } ) $, the matrix representation of $ U^{H,F}_N $ is as follows:
\[
    U^{H,F}_2 =
    \begin{bmatrix}
        O       & A^{H,F}   \\
        A^{H,F} & O
    \end{bmatrix}   
\]
and
\[
    U^{H,F}_N =
    \begin{bmatrix}
        O      & L      & O      & \cdots & O      & R       \\
        R      & O      & L      & \cdots & O      & O       \\
        O      & R      & O      & \cdots & O      & O       \\
        \vdots & \vdots & \vdots & \ddots & \vdots & \vdots  \\
        O      & O      & O      & \cdots & O      & L       \\
        L      & O      & O      & \cdots & R      & O
    \end{bmatrix}   
\]
for $ N \ge 3 $, where $ O $ represents the zero matrix, and $ R $ and $ L $ are the following matrices:
\begin{align*}
    L := \frac{1}{\sqrt{2}}
    \begin{bmatrix}
        1 & 1   \\
        0 & 0
    \end{bmatrix},
    &&
    R := \frac{1}{\sqrt{2}}
    \begin{bmatrix}
        0 & 0   \\
        1 & -1
    \end{bmatrix}.
\end{align*}
Note that $ L $ and $ R $ can be represented as
\begin{align*}
    L = 
    \begin{bmatrix}
        1 & 0 \\
        0 & 0
    \end{bmatrix}
    A^{H,F}
    &&
    \text{and}
    &&
    R = 
    \begin{bmatrix}
        0 & 0 \\
        0 & 1
    \end{bmatrix}
    A^{H,F}.
\end{align*}

Furthermore, let $ f^{H,F}_N $ be the characteristic polynomial of $ U^{H,F}_N $, that is,
\[
    f^{H,F}_N(x) := \det( x I_{2N} - U^{H,F}_N ).
\]
By definition, the factorization of $ f^{H,F}_N(x) $ is obtained:
\begin{prop}\label{Hfactorization}
    For $ N \ge 2 $, it holds that
    \begin{align*}
        f^{H,F}_N(x) &= \prod_{k=0}^{N-1} ( x^2 - \sqrt{2} \cos( 2 \pi k / N )\, x + 1 ).
    \end{align*}
\end{prop}


\subsection{Periods of F-type Hadamard walks}
The periods of M-type Hadamard walks are already known by Dukes \cite{Dukes2014} and Konno et al. \cite{KST2017} The result is as follows:

\begin{thm}
    Let $ T^{H,M}_N $ be the period of M-type Hadamard walk on $ C_N ( N \ge 2 ) $. Then the following holds:
    \[
        T^{H,M}_N =
        \begin{cases}
            2,      & (N=2),    \\
            8,      & (N=4),    \\
            24,     & (N=8),    \\
            \infty, & (\text{otherwise}).
        \end{cases}
    \]
\end{thm}

In this paper, we point out that the same approach as for M type is effective for F type and give the proof below. Now, let $ T^{H,F}_N $ be the period of F-type Hadamard walk on $ C_N $ $ (N\ge2) $.
\begin{thm}
    The periods of F-type Hadamard walk on $ C_N $ $ ( N \ge 2 ) $ are as follows:
    \[
        T^{H,F}_N =
        \begin{cases}
            8,      & (N=2),    \\
            8,      & (N=4),    \\
            24,     & (N=8),    \\
            \infty, & (\text{otherwise}).
        \end{cases}
    \]
\end{thm}
\begin{proof}
    First, by Proposition \ref{Hfactorization}, we know
    \begin{align*}
        f^{H,F}_2(x) &= \Phi_8(x),    \\
        f^{H,F}_4(x) &= \Phi_4(x)^2 \Phi_8(x),   \\
        f^{H,F}_8(x) &= \Phi_3(x)^2 \Phi_4(x)^2 \Phi_6(2)^2 \Phi_8(x).
    \end{align*}
    Thus, $ T^{H,F}_2 = 8 $, $ T^{H,F}_4 = 8 $, and $ T^{H,F}_8 = 24 $ hold.

    Then, if $ N $ has odd prime as its factor, the same argument in \cite{KST2017} can be applied. Therefore, the remainder is the case where $ N $ is a power of $ 2 $ which is greater than $ 2^4 $.
    
    Let $ N = 2^n $.
%
    First, we consider the case where $ n=4 $. We get
    \begin{align*}
        f^{H,F} _{2^4}(x) &= x^{32} + 8 x^{30} + 34 x^{28} + 100 x^{26} + \frac{901}{4} x^{24} + 409 x^{22} + \frac{2465}{4} x^{20} + \frac{1567}{2} x^{18} \\
        & \quad + 848 x^{16} + \frac{1567}{2} x^{14} + \frac{2465}{4} x^{12} + 409 x^{10} + \frac{901}{4} x^{8} + 100 x^{6} + 34 x^{4} + 8 x^{2} + 1.
    \end{align*}
    Thus $ f^{H,F} _{2^4}(x) $ is monic and in $ \mathbf{Q}[x] $ but not in $ \mathbf{Z}[x] $. Then, by Proposition \ref{monic_unity}, we see that $ f^{H,F} _{2^4}(x) $ has a root which is not a root of unity. Therefore, we get $ T^{H,F}_{2^4} = \infty $.

    For $ n > 4 $, $ f^{H,F}_{2^n}(x) $ has $ f^{H,F}_{2^4}(x) $ as a factor. More explicitly, Proposition \ref{Hfactorization} implies
    \begin{align*}
        f^{H,F}_{2^n}(x) &= \prod_{k=0}^{2^n-1} ( x^2 - \sqrt{2} \cos( 2 \pi k / 2^n )\, x + 1 )   \\
        &= \prod_{ \substack{ 0 \le k \le 2^n-1 \\ 2^{n-4} \mid k } } ( x^2 - \sqrt{2} \cos( 2 \pi k / 2^n )\, x + 1 ) \prod_{ \substack{ 0 \le k \le 2^n-1 \\ 2^{n-4} \nmid k } } ( x^2 - \sqrt{2} \cos( 2 \pi k / 2^n )\, x + 1 ) \\
        &= \prod_{ l=0 }^{ 2^4-1 } ( x^2 - \sqrt{2} \cos( 2 \pi \cdot 2^{n-4} l / 2^n )\, x + 1 ) \prod_{ \substack{ 0 \le k \le 2^n-1 \\ 2^{n-4} \nmid k } } ( x^2 - \sqrt{2} \cos( 2 \pi k / 2^n )\, x + 1 )  \\
        &= f^{H,F}_{2^4}(x) \prod_{ \substack{ 0 \le k \le 2^n-1 \\ 2^{n-4} \nmid k } } ( x^2 - \sqrt{2} \cos( 2 \pi k / 2^n )\, x + 1 ).
    \end{align*}
    Therefore, also $ f^{H,F}_{2^n}(x) $ has a root which is not a root of unity. Thus, we have that $ T^{H,F}_{2^n} = \infty $ holds.
\end{proof}

We should remark that the corresponding equation to $ f^{H,M}_{2^4}(x) $ for M-type case is given by
\begin{align*}
    f^{H,M} _{2^4} (x) =& x^{32} - 8 x^{30} + \frac{69}{2} x^{28} - 103 x^{26} + \frac{3761}{16} x^{24} - \frac{1725}{4} x^{22} + \frac{10473}{16} x^{20} - \frac{6687}{8} x^{18} 
\\
&+ 906 x^{16} - \frac{6687}{8} x^{14} + \frac{10473}{16} x^{12} - \frac{1725}{4} x^{10} + \frac{3761}{16} x^{8} - 103 x^{6} + \frac{69}{2} x^{4} - 8 x^{2} + 1.
\end{align*}

\subsection{Absolute zeta functions of zeta functions of Hadamard walks}
Denote the zata functions of M- and F-type Hadamard walks on $ C_N $ by $ \zeta^{H,M}_{C_N} $ and $ \zeta^{H,F}_{C_N} $, respectively $ ( N \ge 2 ) $. For the case where $ N=2,4,8 $, we have the explicit forms of the zeta functions:
\begin{align*}
    \zeta^{H,M}_{C_2}(u) &= \frac{1}{(1-u^2)^2},    \\
    \zeta^{H,M}_{C_4}(u) &= \frac{ u^4-1 }{ (u^2-1)^2 ( u^8-1 ) },    \\
    \zeta^{H,M}_{C_8}(u) &= \frac{ (u^4-1)^3 (u^6-1)^2 }{ (u^2-1)^4 ( u^8-1 ) ( u^{12}-1 )^2 },    \\
    \zeta^{H,F}_{C_2}(u) &= \frac{ u^4-1 }{ u^8-1 },    \\
    \zeta^{H,F}_{C_4}(u) &= \frac{ (u^2-1)^2 }{ ( u^4-1 ) ( u^8-1 ) },    \\
    \zeta^{H,F}_{C_8}(u) &= \frac{ (u^2-1)^4 }{ ( u^4-1 ) ( u^6-1 )^2 ( u^8-1 ) }.
\end{align*}
Note that these are absolute automorphic forms of weight $ -4 $, $ -8 $, $ -16 $, $ -4 $, $ -8 $, and $ -16 $, respectively.

From the discussion so far, we can deduce the next theorem about these functions. In particular, we can get explicit expressions of the absolute zeta functions of these zeta functions.

\begin{thm}\label{HAZ}
    We have the following explicit expressions of absolute zeta functions and their functional equations.

    M type: For $ N=2 $,
    \begin{align*}
        Z_{\zeta^{H,M}_{C_2}}(w, s) 
        &= \zeta_{2} \left(w, s + 4, (2,2) \right),
        \\
        \zeta_{\zeta^{H,M}_{C_2}}(s)
        &= \Gamma_{2} \left( s + 4, (2,2) \right),
        \\
        \zeta_{\zeta^{H,M}_{C_2}}(-4-s) 
        &= S_{2} \left( s + 4, (2,2) \right) \ \zeta_{\zeta^{H,M}_{C_2}}(s).
    \end{align*}

    For $ N=4 $,
    \begin{align*}
        Z_{\zeta^{H,M}_{C_4}}(w, s) 
        &= \sum_{I \subset \{1 \}} (-1)^{|I|} \ \zeta_{3} \left(w, s + 8 + 4|I|), (2,2,8) \right),
        \\
        \zeta_{\zeta^{H,M}_{C_4}}(s)
        &= \prod_{I \subset \{1 \}} \Gamma_{3} \left( s + 8 + 4|I|, (2,2,8) \right)^{ (-1)^{|I|}},
        \\
        \zeta_{\zeta^{H,M}_{C_4}} (-8-s)
        &= \Bigl( \prod_{I \subset \{1 \}} S_{3} \left(s + 8 + 4|I|, (2,2,8) \right)^{ (-1)^{|I|}} \Bigr) \zeta_{\zeta^{H,M}_{C_4}} (s).
    \end{align*}

    For $ N=8 $,
    \begin{align*}
        Z_{\zeta^{H,M}_{C_8}} (w, s) 
        &= \sum_{I \subset \{1, \ldots, 5 \}} (-1)^{|I|} \ \zeta_{7} \left(w, s + 16 + m(I), (2,2,2,2,8,12,12) \right),
        \\
        \zeta_{\zeta^{H,M}_{C_8}} (s)
        &= \prod_{I \subset \{1, \ldots, 5 \}} \Gamma_{7} \left( s + 16  + m(I), (2,2,2,2,8,12,12) \right)^{ (-1)^{|I|}},
        \\
        \zeta_{\zeta^{H,M}_{C_8}} (-16-s) 
        &= \Bigl( \prod_{I \subset \{1, \ldots, 5 \}} S_{7} \left( s + 16  + m(I), (2,2,2,2,8,12,12) \right)^{ (-1)^{|I|}} \Bigr) \zeta_{\zeta^{H,M}_{C_8}} (s).
    \end{align*}

    F type: For $ N=2 $,
    \begin{align*}
        Z_{\zeta_{C_2}^{H,F}} (w, s) 
        &= \zeta_{1} \left(w, s + 4, (8) \right) - \zeta_{1} \left(w, s + 8, (8) \right),
        \\
        \zeta_{\zeta_{C_2}^{H,F}} (s)
        &= \frac{\Gamma \left( \frac{s+4}{8} \right)}{\Gamma \left( \frac{s+8}{8} \right)} \cdot n^{-\frac{1}{2}},
        \\
        \zeta_{\zeta_{C_2}^{H,F}} (-4-s) 
        &= - \cot \left( \frac{s \pi}{8} \right) \zeta_{\zeta_{C_2}^{H,F}} (s).
    \end{align*}

    For $ N=4 $,
    \begin{align*}
        Z_{\zeta_{C_4}^{H,F}} (w, s) 
        &= \sum_{I \subset \{1, 2 \}} (-1)^{|I|} \ \zeta_{2} \left(w, s + 8 + 2|I|), (4,8) \right),
        \\
        \zeta_{\zeta_{C_4}^{H,F}} (s)
        &= \prod_{I \subset \{1, 2 \}} \Gamma_{2} \left( s + 8 + 2|I|, (4,8) \right)^{ (-1)^{|I|}},
        \\
        \zeta_{\zeta_{C_4}^{H,F}} (-8-s) 
        &= \Bigl( \prod_{I \subset \{1,2 \}} S_{2} \left(s + 8 + 2|I|, (4,8) \right)^{ (-1)^{|I|}} \Bigr) \zeta_{\zeta_{C_4}^{H,F}}(s).
    \end{align*}

    For $ N=8 $,
    \begin{align*}
        Z_{\zeta_{C_8}^{H,F}} (w, s) 
        &= \sum_{I \subset \{1, \ldots, 4 \}} (-1)^{|I|} \ \zeta_{4} \left(w, s + 16 + m(I), (4,6,6,8) \right),
        \\
        \zeta_{\zeta_{C_8}^{H,F}} (s)
        &= \prod_{I \subset \{1, \ldots, 4 \}} \Gamma_{4} \left( s + 16  + m(I), (4,6,6,8) \right)^{ (-1)^{|I|}},
        \\
        \zeta_{\zeta_{C_8}^{H,F}} (-16-s) 
        &= \prod_{I \subset \{1, \ldots, 4 \}} S_{4} \left( s + 16  + m(I), (4,6,6,8) \right)^{ (-1)^{|I|}} \zeta_{\zeta_{C_8}^{H,F}(s)}.
    \end{align*}
\end{thm}

\begin{proof}
    The proofs are almost the same for each $ \zeta^{H,M}_{C_N} $ and $ \zeta^{H,F}_{C_N} $. Thus we prove the theorem only for $ \zeta^{H,M}_{C_4} $. First, by definition,
    \[
        \zeta^{H,M}_{C_4}(u) = \det\Bigl( I_8 - u U^{H,M}_4 \Bigr)^{-1}.
    \]
    Now, by Proposition \ref{Hfactorization}, we know
    \[
        \det\Bigl( x I_8 - U^{H,M}_4 \Bigr) = \frac{ (x^2-1)^2 ( x^8-1 ) }{ x^4-1 }.
    \]
    By substituting $ x = 1/u $, we get
    \begin{align*}
         \det\Bigl( \frac{1}{u} I_8 - U^{H,M}_4 \Bigr) &= \frac{ ((1/u)^2-1)^2 ( (1/u)^8-1 ) }{ (1/u)^4-1 },\\
         \det\Bigl( I_8 - u U^{H,M}_4 \Bigr) &= \frac{ (u^2-1)^2 ( u^8-1 ) }{ u^4-1 }.
    \end{align*}
    Therefore, we conclude
    \[
        \zeta^{H,M}_{C_4}(u) = \frac{ u^4-1 }{ (u^2-1)^2 ( u^8-1 ) }.
    \]
    Here, we see that $ \zeta^{H,M}_{C_4} $ is an absolute automorphic form of weight $ -8 $. Then, by Theorem \ref{ExplicitAZeta}, we obtain the desired result.
\end{proof}

\section{Grover walks}\label{Grover}

Grover walks are another well-studied class of quantum walks. There are two types of Grover walks as in the case of Hadamard walks: M- and F-type. They are characterized by local coin operators as usual. In this paper, we focus on the $3$-states model. First, we introduce the definition:

\begin{defi}
    \emph{M-type} (respectively \emph{F-type}) \emph{Grover walks with $3$ states} on $ C_N $ are quantum walks whose coin operators with respect to the ordered basis $ ( \ket{\leftarrow}, \ket{\cdot}, \ket{\rightarrow} ) $ are as follows respectively:
    \begin{align*}
        A^{G_3,M} = \frac{1}{ 3 }
        \begin{bmatrix}
            -1   &   2  &   2   \\
             2   &  -1  &   2   \\
             2   &   2  &  -1
        \end{bmatrix},
        &&  A^{G_3,F} = \frac{1}{ 3 }
        \begin{bmatrix}
             2   &   2  &  -1   \\
             2   &  -1  &   2   \\
            -1   &   2  &   2
        \end{bmatrix}.
    \end{align*}
\end{defi}

We write $ U^{G_3,M}_N $ and $ U^{G_3,F}_N $ for the time-evolution opeartors of M- and F-type Grover walks on $ C_N $, respectively. With respect to the ordered basis of $ \mathcal{H} $ $ ( \bra{0} \otimes \bra{\leftarrow},\allowbreak \bra{0} \otimes \bra{\cdot}, \bra{0} \otimes \bra{\rightarrow}, \bra{1} \otimes \bra{\leftarrow}, \dots, \bra{N-1} \otimes \bra{ \rightarrow } ) $, the matrix representation of $ U^{G_3,M}_N $ and $ U^{G_3,F}_N $ is as follows:
\[
    U^{G_3,X}_2 =
    \begin{bmatrix}
        S   & L^{(X)} + R^{(X)}   \\
        L^{(X)} + R^{(X)} & S
    \end{bmatrix},
    \qquad X=M,F
\]
and
\[
    U^{G_3,X}_N =
    \begin{bmatrix}
        S       & L^{(X)} & O       & \cdots & O       & R^{(X)} \\
        R^{(X)} & S       & L^{(X)} & \cdots & O       & O       \\
        O       & R^{(X)} & S       & \cdots & O       & O       \\
        \vdots  & \vdots  & \vdots  & \ddots & \vdots  & \vdots  \\
        O       & O       & O       & \cdots & S       & L^{(X)} \\
        L^{(X)} & O       & O       & \cdots & R^{(X)} & S
    \end{bmatrix},
    \qquad X=M,F 
\]
for $ N \ge 3 $, where $ O $ represents the zero matrix, and $ S, L^{(X)}, $ and $ R^{(X)} $ are the following matrices:
\[
    S :=
    \begin{bmatrix}
        0 & 0 & 0   \\
        0 & 1 & 0   \\
        0 & 0 & 0
    \end{bmatrix}
    A^{G_3,M}
    =
    \begin{bmatrix}
        0 & 0 & 0   \\
        0 & 1 & 0   \\
        0 & 0 & 0
    \end{bmatrix}
    A^{G_3,F},
\]
\begin{align*}
    L^{(X)} :=
    \begin{bmatrix}
        1 & 0 & 0   \\
        0 & 0 & 0   \\
        0 & 0 & 0
    \end{bmatrix}
    A^{G_3,X},
    &&
    R^{(X)} :=
    \begin{bmatrix}
        0 & 0 & 0   \\
        0 & 0 & 0   \\
        0 & 0 & 1
    \end{bmatrix}
    A^{G_3,X},
    \qquad X=M,F.
\end{align*}

Furthermore, let $ f^{G_3,M}_N $ and $ f^{G_3,F}_N $ be the characteristic polynomials of $ U^{G_3,M}_N $ and $ U^{G_3,F}_N $, respectively:
\begin{align*}
    f^{G_3,M}_N(x) := \det( x I_{3N} - U^{G_3,M}_N ),   &&  f^{G_3,F}_N(x) := \det( x I_{3N} - U^{G_3,F}_N ).
\end{align*}
By definition, the factorization of $f^{G_3,M}_N(x)$ and $f^{G_3,F}_N(x)$ are obtained:
\begin{prop}\label{Gfactorization}
    For $ N \ge 2 $, it holds that
    \begin{align*}
        f^{G_3,M}_N(x) &= ( x - 1 )^N \prod^{N-1}_{k =0} \biggl( x^2 + \frac{2}{3} \biggl( 2 + \cos \left( \frac{2 \pi k}{N} \right) \biggr) \ x + 1 \biggr),   \\
        f^{G_3,F}_N(x) &= ( x - 1 )^N \prod^{N-1}_{k =0} \biggl( x^2 - \frac{2}{3} \biggl( 2 + \cos \left( \frac{2 \pi k}{N} \right) \biggr) \ x + 1 \biggr).
    \end{align*}
\end{prop}
See \cite{KKKS} for this proposition, for example. In fact, the factorization formula is shown only for M-type in \cite{KKKS}, but we can deduce for F-type by a similar way.

\subsection{Periods of Grover walks with 3 states}
The periods of Grover walks with 3 states were clarified by Kajiwara et al. \cite{KKKS} like the following:
\begin{thm}
    Let $ T^{G_3,M}_N $ be the period of M-type Grover walk with $3$ states on $ C_N $ $ (N\ge2) $. Then
    \[
        T^{G_3,M}_N =
        \begin{cases}
            6,      & (N=3),    \\
            \infty, & (\text{otherwise})
        \end{cases}
    \]
    holds.

    Similarly, the period of F-type Grover walk with $3$ states on $ C_N $ $ T^{G_3,F}_N $ $ (N\ge2) $  is calculated as
    \[
        T^{G_3,F}_N =
        \begin{cases}
            4,      & (N=3),    \\
            \infty, & (\text{otherwise}).
        \end{cases}
    \]
\end{thm}

In this paper, we point out that the same approach as in the proof in the previous subsection works for this theorem. The proof of this approach is given below.

\begin{proof}
    First, we consider the M type.
    
    For $ N=3 $, we have
    \[
        f^{G_3,M}_3(x) = \Phi_1(x)^3 \Phi_2(x)^2 \Phi_3(x)^2
    \]
    by Proposition \ref{Gfactorization}. The result is a direct consequence of this factorization.

    For $ N=2 $, by Proposition \ref{Gfactorization} again, we see
    \[
        f^{G_3,M}_2(x) = x^6 + \frac{2}{3} x^5 - x^4 - \frac{4}{3} - x^2 + \frac{2}{3} + 1.
    \]
    Thus, $ f^{G_3,M}_2(x) $ is monic and in $ \mathbf{Q}[x] $ but not in $ \mathbf{Z}[x] $. Therefore, by Proposition \ref{monic_unity}, $ f^{G_3,M}_N(x) $ has a root which is not a root of unity, and this implies $ T^{G_3,M}_2 = \infty $.

    In the case of $ N \ge 4 $ which is not a multiple of $3$, we get
    \[
        f^{G_3,M}_N(x) = \prod_{k=0}^{N-1} \Bigl( x^3 + \frac{1}{3} \Bigl( 1 + 2 \cos\Bigl( \frac{ 2 \pi k }{ N } \Bigr) \Bigr) x^2 - \frac{1}{3} \Bigl( 1 + 2 \cos\Bigl( \frac{ 2 \pi k }{ N } \Bigr) \Bigr) x - 1 \Bigr)
    \]
    by Proposition \ref{Gfactorization} again. Then the coefficient of $ x^{3N-1} $ can be calculated:
    \begin{align*}
        \sum_{k=0}^{N-1} \frac{1}{3} \Bigl( 1 + 2 \cos\Bigl( \frac{ 2 \pi k }{ N } \Bigr) \Bigr) &= \frac{N}{3} + \sum_{k=0}^{N-1} \cos\Bigl( \frac{ 2 \pi k }{ N } \Bigr)  \\
        &= \frac{N}{3}.
    \end{align*}
    Therefore, $ f^{G_3,M}_N(x) $ is not in $ \mathbf{Z}[x] $, and $ f^{G_3,M}_N(x) $ is monic and in $ \mathbf{Q}[x] $ by definition. Thus, by Proposition \ref{monic_unity}, we know $ T^{G_3,M}_N = \infty $.

    In the case of $ N \ge 4 $ which is a multiple of $ 3 $ but not a power of $ 3 $, we can take a prime factor $p$ of $N$ such that $ p \neq 3 $. Then like the above-mentioned proof in the previous subsection, we can show that $ f^{G_3,M}_N(x) $ has $ f^{G_3,M}_p(x) $ as its factor. Here, we already know that $ f^{G_3,M}_p(x) $ has a root which is not a root of unity. Thus, also $ f^{G_3,M}_N(x) $ has a root which is not a root of unity, and we get $ T^{G_3,M}_N = \infty $.

    Finally, in the case of $ N \ge 4 $ which is a power of $ 3 $, it is enough to examine $ f^{G_3,M}_9(x) $, because $ f^{G_3,M}_N(x) $ has $ f^{G_3,M}_9(x) $ as its factor. By Proposition \ref{Gfactorization}, we have
    \begin{align*}
        f^{G,M} _{9} (x) 
        &= x^{27} + 3 x^{26} - \frac{128}{9} x^{24} - \frac{214}{9} x^{23} + \frac{62}{9} x^{22} + \frac{5752}{81} x^{21} + \frac{6376}{81} x^{20} 
        \\
        & - \frac{3331}{81} x^{19} - \frac{15059}{81} x^{18} - \frac{11686}{81} x^{17} + \frac{8728}{81} x^{16} +\frac{23752}{81} x^{15} + \frac{4316}{27} x^{14} 
        \\
        & - \frac{4316}{27} x^{13} - \frac{23752}{81} x^{12} - \frac{8728}{81} x^{11} + \frac{11686}{81} x^{10} + \frac{15059}{81} x^{9} + \frac{3331}{81} x^{8} 
        \\
        & - \frac{6376}{81} x^{7} - \frac{5752}{81} x^{6} - \frac{62}{9} x^{5} + \frac{214}{9} x^{4} + \frac{128}{9} x^{3} - 3 x - 1.
    \end{align*}
    Thus, $ f^{G_3,M}_9(x) $ is monic and in $ \mathbf{Q}[x] $ but not in $ \mathbf{Z}[x] $. Therefore, by Proposition \ref{monic_unity}, $ f^{G_3,M}_N(x) $ has a root which is not a root of unity, and this implies $ T^{G_3,M}_N = \infty $.

    Then the same method can be used to prove the case of F type. Note that we can obtain the following expanded form of $ f^{G_3,F}_9(x) $:
    \begin{align*}
        f^{G_3,F} _{9} (x) 
        &= x^{27} + 3 x^{26} + 3 x^{25} + \frac{25}{9} x^{24} + \frac{26}{9} x^{23} - \frac{2}{9} x^{22} + \frac{46}{81} x^{21} 
        \\
        & + \frac{106}{81} x^{20} - \frac{59}{27} x^{19} + \frac{125}{81} x^{18} - \frac{1}{27} x^{17} - \frac{353}{81} x^{16} - \frac{116}{81} x^{15} - \frac{212}{27} x^{14} 
        \\
        & - \frac{212}{27} x^{13} - \frac{116}{81} x^{12} - \frac{353}{81} x^{11} - \frac{1}{27} x^{10} + \frac{125}{81} x^{9} - \frac{59}{27} x^{8} 
        \\
        & + \frac{106}{81} x^{7} + \frac{46}{81} x^{6} - \frac{2}{9} x^{5} + \frac{26}{9} x^{4} + \frac{25}{9} x^{3} + 3 x^2 + 3 x + 1.
    \end{align*}
\end{proof}

\subsection{Absolute zeta functions of zeta functions of Grover walks with 3 states}

\begin{thm}
    We have the following explicit expressions of absolute zeta functions and their functional equations.

    M type:
    \begin{align*}
        Z_{\zeta_{C_3}^{G_3,M}} (w, s) 
        &= - \sum_{I \subset \{1 \}} (-1)^{|I|} \ \zeta_{4} \left(w, s + 9  + |I|, (2,2,3,3) \right),
        \\
        \zeta_{\zeta_{C_3}^{G,M}} (s)
        &= \prod_{I \subset \{1 \}} \Gamma_{4} \left( s + 9  + |I|, (2,2,3,3) \right)^{ (-1)^{|I|+1}},
        \\
        \zeta_{\zeta_{C_3}^{G,M}} (-9-s)^{-1} 
        &= \Bigl( \prod_{I \subset \{1 \}} S_{4} \left( s + 9  + |I|, (2,2,3,3) \right)^{ (-1)^{|I|+1}} \Bigr) \zeta_{\zeta_{C_3}^{G,M}} (s).
    \end{align*}

    F type:
    \begin{align*}
        Z_{\zeta_{C_3}^{G_3,F}} (w, s) 
        &= - \sum_{I \subset \{1 \}} (-1)^{|I|} \ \zeta_{3} \left(w, s + 9  + |I|, (2,4,4) \right),
        \\
        \zeta_{\zeta_{C_3}^{G_3,F}} (s)
        &= \prod_{I \subset \{1 \}} \Gamma_{3} \left( s + 9  + |I|, (2,4,4) \right)^{ (-1)^{|I|+1}},
        \\
        \zeta_{\zeta_{C_3}^{G_3,F}} (-9-s) 
        &= \Bigl( \prod_{I \subset \{1 \}} S_{3} \left( s + 9  + |I|, (2,4,4) \right)^{ (-1)^{|I|+1}} \Bigr) \zeta_{\zeta_{C_3}^{G_3,F}} (s).
    \end{align*}
\end{thm}
\begin{proof}
    The proof is almost the same as that of Theorem \ref{HAZ}. Remark that the explicit forms of the zeta functions are as follows:
    \begin{align*}
        \zeta^{G_3,M}_{C_3}(u) &= \frac{u-1}{ (u^2-1)^2 (u^3-1)^2 },   \\
        \zeta^{G_3,F}_{C_3}(u) &= \frac{ u-1 }{ (u^2-1) (u^4-1)^2 }.
    \end{align*}
    Moreover, both of them are absolute automorphic forms of weight $ -9 $.
\end{proof}

\section{Conclusion}\label{conclusion}
In this paper, we dealt with the periods of some kind of quantum walks on cycle graphs and absolute zeta functions of their zeta functions. The quantum walks we studied in this paper are Hadamard walks and Grover walks with 3 states. We provided a simple and unified proof for these periods using the notion of cyclotomic polynomials. Moreover, we pointed out that it is possible to compute the absolute zeta functions of such quantum walks by applying a theorem of Kurokawa et al. \hspace{-4pt}(Theorem \ref{ExplicitAZeta}) if the period is finite. Also, we obtained some functional equations of such absolute zeta functions by the same theorem.



\par
\
\par
\noindent


\end{document}